\documentclass{article}
\usepackage{mdframed}
\usepackage[most]{tcolorbox}
\usepackage{todonotes}
\usepackage{PRIMEarxiv}
\usepackage[utf8]{inputenc} 
\usepackage[T1]{fontenc}    
\usepackage{hyperref}  
\usepackage{amssymb,amsmath,amsthm}
\usepackage{xspace}
\usepackage{mathtools}
\usepackage{graphicx,xcolor}
\usepackage{url}            
\usepackage{booktabs}       
\usepackage{amsfonts}       
\usepackage{nicefrac}       
\usepackage{microtype}      
\usepackage{lipsum}
\usepackage{fancyhdr}       
\usepackage{graphicx} 
\usepackage[ruled,linesnumbered]{algorithm2e}
\usepackage[utf8]{inputenc}
\graphicspath{{media/}}     

\newcommand{\np}{\textup{\textsf{NP}}\xspace}

\newcommand{\wth}{\textsf{W[2]}\textrm{-hard}\xspace}

\newcommand{\mcs}{\textup{\textsc{MCS}}\xspace}
\newcommand{\ptas}{\textup{\textsc{PTAS}}\xspace}
\newcommand{\acs}{\textup{\textsc{ACS}}\xspace}
\newcommand{\cs}{\textup{\textsc{CS}}\xspace}

\newcommand{\NN}{\mathrm{NN}\xspace}

\newcommand{\cov}{\mathrm{COV}\xspace}

\newcommand{\apx}{\textup{\textsf{APX}}\xspace}
\newcommand{\NP}{\textup{\textsf{NP}}\xspace}
\newcommand{\PP}{\textup{\textsf{P}}\xspace}

\newcommand{\ath}{\mathrm{PATH}\xspace}
\newcommand{\dist}{\mathrm{d}\xspace}
\newcommand{\neigh}{\mathrm{N}\xspace}
\newcommand{\xp}{{\sf XP}\xspace}
\newcommand{\fpt}{{\sf FPT}\xspace}

\usepackage{lineno}

\newtheorem{theorem}{Theorem}[section]

\newtheorem{lemma}[theorem]{Lemma}
\newtheorem{definition}{Definition}[section]
\newtheorem{observation}{Observation}[theorem]
\title{Minimum Consistent Subset in Interval Graphs and Circle Graphs
}
\author{
  Bubai Manna\\
  Department of Mathematics, Indian Institute of Technology Kharagpur, India. \\
  \texttt{bubaimanna11@gmail.com}
}
\begin{document}
\maketitle
\begin{abstract}
\noindent In a connected simple graph $G = (V,E)$, each vertex of $V$ is colored by a color from the set of colors $C=\{c_1, c_2,\dots, c_{\alpha}\}$. We take a subset $S$ of $V$, such that for every vertex $v$ in $V\setminus S$, at least one vertex of the same color is present in its set of nearest neighbors in $S$. We refer to such a $S$ as a consistent subset (\cs). The Minimum Consistent Subset (\mcs) problem is the computation of a consistent subset of the minimum size. It is established that \mcs is NP-complete for general graphs, including planar graphs. We expand our study to interval graphs and circle graphs in an attempt to gain a complete understanding of the computational complexity of the \mcs problem across various graph classes.

\smallskip

\noindent This work introduces an $(4\alpha+ 2)$- approximation algorithm for \mcs in interval graphs where $\alpha$ is the number of colors in the interval graphs. Later, we show that in circle graphs, \mcs is \apx-hard.

\vspace{-0.1in}
\end{abstract}


\section{Introduction}\label{intro1}
\noindent A computer (or person) can utilize a method to classify the discrimination between two or more classes of objects in order to provide the system with a number of sample objects and the right classifications for each. The system may accept a large number of objects as input and identify the class of each sample object. Such systems have been developed for a wide range of essential applications, including speech recognition, handwriting recognition, and character recognition. As a result, the methods utilized to achieve such systems differ in terms of how examples are presented and how the likely classification of a sample is determined. However, there is a concern with the depiction of the objects. Assume the objects are points in a 2-dimensional metric.

\noindent Many supervised learning approaches use a colored training dataset $T$ in a metric space $(X, d)$ as input, with each element $t \in T$ assigned a color $c_i$ from the set of colors $C=\{c_1, c_2, \dots, c_{\alpha}\}$. The purpose is to preprocess $T$ in order to meet specific optimality requirements by immediately assigning color to any uncolored element in $X$. The closest neighbor rule, which assigns a color to each uncolored element $x$ based on the colors of its $k$ closest neighbors in the training dataset $T$ (where $k$ is a fixed integer), is a popular optimality criterion. The efficiency of such an algorithm is determined by the size of the training dataset. As a result, it is necessary to reduce the size of the training set while retaining the distance data. Hart \cite{Hart} introduced this concept in $1968$ with the minimum consistent subset (\mcs) problem. Given a colored training dataset $T$, the goal is to find a subset $S \subseteq T$ of the minimum cardinality such that the color of every point $t \in T$ is the same as one of its nearest neighbors in $S$. Over $2800$ citations to \cite{Hart} on Google Scholar demonstrate that the \mcs problem has found several uses since its inception. The article \cite{Wilfong} demonstrated that the \mcs problem for points in $\Re^2$ under the Euclidean norm is \np-complete if at least three colors color the input points. Furthermore, it is \np-complete even with two colors \cite{Bodhayan}. Recently, in the article \cite{Chitnis}, it was established that the \mcs problem is W[1]-hard when parameterized by the output size. Also, the particles \cite{Wilfong} and \cite{Biniaz} offer algorithms for the \mcs problem in $\Re^2$.

\subsection{Notaion, Definition and Preliminary Results}\label{intro} \noindent This paper explores the minimum consistent subset problem when $(\mathcal{X},d)$ is a graph metric. Most graph theory symbols and notations are standard and come from \cite{diestel2012graph}. For every graph $G$, we refer to the set of vertices as $V$ or $V(G)$, and the set of edges as $E$ or $E(G)$. Consider a graph $G$ with a vertex coloring function $C:V(G) \rightarrow \{c_1, c_2, \dots, c_{\alpha}\}$. Let $C(U)$ represent the set of colors of the vertices in $U$. Formally, $C(U)=\{C(u):u\in U\}$. For any two vertices $u,v \in V(G)$, we use $\dist(u,v)$ to signify the shortest path distance (number of edges) between $u$ and $v$ in $G$. $\dist(v,U)=\min_{u\in U}\dist(v,u)$ for a vertex $v\in V(G)$ and any subset of vertices $U\subseteq V$. For any graph $G$ and any vertex $v\in V(G)$, let $\neigh(v)=\{u\mbox{ }|\mbox{ }u\in V(G), (u,v)\in E(G)\}$ indicate the set of neighbors of $v$ and $\neigh[v]=\{v\}\cup \neigh(v)$. The distance between two subgraphs $G_1$ and $G_2$ in $G$ is represented as $\dist(G_1,G_2)=\min \{\dist(v_1,v_2)\mbox{ }|\mbox{ }v_1\in V(G_1),v_2\in V(G_2)\}$. For any subset of vertices $U\subseteq V$ in a graph $G$, $G[U]$ represents the subgraph of $G$ induced on $U$. The nearest neighbor of $v$ in the set $U$ is indicated as $\NN(v,U)$, which is formally defined as $\{u\in U \mbox{ }| \mbox{ }\dist(v,u)=\dist(v,U)\}$. Also, for $X, U\subseteq V$, $\NN(X,U)=\cup_{v\in X}\{u\in U \mbox{ }| \mbox{ }\dist(v,u)=\dist(v,U)\}$. The shortest path in $G$ is $\ath(v,u)$, which connects vertices $v$ and $u$. $\ath(v,U)=\{\ath(v,u) \mbox{ }|\mbox{ }u\in \NN(v,U)\}$. 
 
\smallskip

\noindent If there is a vertex $u\in \NN(v, U)$ such that $C(v)=C(u)$, then $v$ is said to be \emph{covered} by the set $U$. We sometimes use only $u$ covers $v$ where $u\in U$. $\cov(v,U)$ refers to the set of vertices in $U$ that cover $v$. $\cov(v,U)=\{u\in U \mbox{ }| \mbox{ }u\in \NN(v,U)\mbox{ and }C(u)=C(v)\}$, and for a subset $X\subseteq V$, $\cov(X,U)=\cup_{v\in X}\{u\in U \mbox{ }| \mbox{ }u\in \NN(v,U)\mbox{ and }C(u)=C(v)\}$. Thus, $u\in U$ covers $v$ indicates that $u\in \cov(v,U)$. If $\cov(v, U)=\emptyset$, then $v$ is not covered by the set $U$; otherwise, $v$ is covered by the set $U$. Similarly, if $C(X)=C(\cov(X,U))$, then $X$ is covered by $U$. If $v$ is not covered by the set $U$, $C(v) \neq C(\cov(v,U))$. Without loss of generality, we shall use $[n]$ to denote the set of integers $\{1, \ldots, n\}$. A subset $S$ is said to be a \emph{consistent subset} (in short \cs) for $(G,C)$ if every $v\in V(G)$, $C(v)\in C(\NN(v,S))$. The consistent subset problem in graphs is defined as follows:

\begin{tcolorbox}[enhanced,title={\color{black} \sc{\mcs on Graphs}}, colback=white, boxrule=0.4pt, attach boxed title to top center={xshift=-2cm, yshift*=-2mm}, boxed title style={size=small,frame hidden,colback=white}]
		
    \textbf{Input:} A graph $G$, a coloring function $C:V(G)\rightarrow \{c_1, c_2, \dots, c_{\alpha}\}$, and an integer $\ell$.\\
    \textbf{Question:} Does there exist a consistent subset of size $\le \ell$ for $(G,C)$?
\end{tcolorbox}
\noindent Figure \ref{fig1}(F) shows an example of both \cs and \mcs. Minimum consistent subset (\mcs) refers to a consistent subset with a minimum size. Banerjee et al.~\cite{Banerjee} demonstrated that the \mcs is \wth when parameterized by the size of the minimum consistent set, even with only two colors. Under basic complexity-theoretic assumptions for generic networks, a \fpt algorithm parameterized by $(c+\ell)$ is not feasible. This naturally raises the question of identifying nontrivial but simple graph classes for which the problem remains computationally intractable. Dey et al.~\cite{Sanjana} proposed a polynomial-time algorithm for \mcs on bi-colored paths, spiders, combs, etc. Also, Dey et al.~\cite{Anil} proposed a polynomial-time algorithm for bi-colored trees. Arimura et al.~\cite{Arunima} introduced a \xp algorithm parameterized by the number of colors $c$, with a running time of $\mathcal{O}(2^{4c}n^{2c+3})$. Recently, we also showed that \mcs on trees is \np-complete in \cite{Bubai}. 

\subsection{New Results}
Only the polynomial time algorithm of \mcs on simple graphs, such as paths, combs, and spiders, are discussed in \cite{Sanjana}, hence the problem is not $\np$-hard on the simple graphs presented in \cite{Sanjana}. However, \cite{Anil} described a polynomial time algorithm of \mcs on bi-chromatic trees. Recently, \cite{Arunima} demonstrated a fixed-parameter tractable of \mcs on trees using the number of colors as a parameter. The article \cite{bubai1} demonstrates that fixed-parameter tractability has a faster running time for the same problem. Aside from these results, no approximation algorithm has been created for $\np$-hard problems on \mcs till now. Therefore, an approximation algorithm for $\np$-hard problems is very much needed. We are looking for intersection graphs, specifically interval graphs and circle graphs. It is proved that the \mcs problem is $\np$-hard on interval graphs in the paper \cite{bubai1}. As a result, we are discovering the first approximation algorithm for a $\np$-hard problem. To discover an approximation, we determine whether a polynomial time approximation scheme (\ptas) is possible for this problem on interval graphs or not. To find this, we can show that \mcs on interval graph can be \apx-complete by using the similar reduction of $\np$-hard presented in \cite{bubai1}. Hence, no \ptas can be formed for the interval graphs unless $\PP=\NP$. Therefore, it is very important to find an approximation algorithm for \mcs in Interval graphs. In section \ref{section3}, we demonstrate the $(4\alpha+2)$-approximation for \mcs in interval graphs. In particular, if $\alpha=2$, this approach provides a $ 10$-approximation algorithm in bi-colored interval graphs. In section \ref{section4}, we prove that \mcs on circle graphs is $\apx$-hard, which means there is no \ptas for \mcs on circle graphs unless $\PP=\np$.

\smallskip
\noindent The definition of interval graphs is as follows.
\begin{definition}
\begin{mdframed}[backgroundcolor=red!10,topline=false,bottomline=false,leftline=false,rightline=false]\label{def3} 
A graph $I$ is said to be an \emph{interval graph} if there exists an interval layout of the graph $I$, or in other words, for each node, $v_i \in V(I)$ one can assign an interval $\alpha_i$ on the real line such that $(v_i,v_j) \in E(I)$ if and only if $\alpha_i$ and $\alpha_j$ (completely or partially) overlap in the layout of those intervals. Such $I$ is referred to as a \emph{interval graph}. Each interval of $I$ now receives a color from the collection $C=\{c_1,c_2,\dots,c_{\alpha}\}$. We sometimes refer to an interval as a vertex. If $I$ has $n$ intervals and $m$ intersects, then $\lvert V(I)\rvert =n$ and $\lvert E(I)\rvert =m$.
\end{mdframed}
\end{definition}
\begin{observation}\label{basicobservation}
The consistent subset must have at least one vertex from each color. If no vertices of color $c_i$ appear in a consistent subset, then vertices of color $c_i$ cannot be covered.
\end{observation}
\begin{observation}\label{basi1}
If $S$ is a consistent subset, it is not necessary for a subset of $S$ to be a \mcs. In Figure \ref{fig1}(F), $\{v_3, v_4, v_8\}$ is a \cs, whereas $\{v_7, v_8\}$ is a \mcs. However, no subset of $\{v_3, v_4, v_8\}$ is \mcs.
\end{observation}

\section{$(4\alpha +2)$-Approximation Algorithm of \mcs for Interval Graphs}\label{section3}

\noindent The following is the essential concept underlying our approximation approaches. A \emph{leaf bar cover} is a simple construction derived from any \mcs. A leaf bar cover made from a \cs cannot be more than twice the size of the corresponding \cs. We then show that using a dynamic programming technique, we can construct a minimum-size leaf bar cover in polynomial time. We then show that any leaf bar cover can be extended to a \cs with little extra cost. Therefore, the method for the $(4\alpha+2)$-approximation algorithm is as follows.

\smallskip

\noindent Let $I$ be an interval graph and $\lvert V(I)\rvert =n$. Also, let $1,2,\dots,2n$ be a left-to-right sequence of interval endpoints in $V(I)$ and it is denoted by $P(I)$. We employ two artificial points, $0$ and $2n+1$, located on the left and right sides of the points $1$ and $2n$, respectively. Additionally, let $C=\{c_1, c_2, \dots, c_{\alpha}\}$ be the collection of colors. For each interval $\ell \in V(I)$, the endpoint with the smaller (or bigger) label is called the \emph{left endpoint} (or \emph{right endpoint}), represented by $l(\ell)$ (or $r(\ell)$). We define $I[j]=\ell$ if $l(\ell)=j$ or $r(\ell)=j$; that is, $I[j]$ is an interval of $V(I)$ with one endpoint (left or right) equal to $j$. For a set of intervals $S\subseteq V(I)$, let $P(S)$ represent the set of all endpoints of the intervals of $S$. For any points $i,j\in P(I)$, let $(i..j)$ be a collection of points obtained by starting from $i$ and going left to right until reaching $j$. We assume that $(i..j)$ is an open set, which means it does not contain the endpoints $i$ and $j$. $(i..j)$ is known as a \emph{bar} (refer to Figure \ref{fig1}(A). We will use $[i..j]$ to represent a bar with endpoints $i$ and $j$. For each bar $s=(i..j)$, we call the endpoint with the smaller (or larger) label the left (or right) endpoint of $s$, denoted by $l(s)$ (or $r(s)$). For any bar $s$, let $I_s$ denote the set of intervals in $V(I)$ with at least one endpoint in $s$, and $O_s=V(I)\setminus I_s$ denote the set of intervals in $V(I)$ with both endpoints outside $s$.

\begin{mdframed}[backgroundcolor=red!10,topline=false,bottomline=false,leftline=false,rightline=false]\label{def5} 
\begin{definition}
Any bar $s$ is a leaf bar if and only if $O_s$ covers $I_s$.
\end{definition}
\end{mdframed}
\begin{observation}\label{easyobserve}
If $S$ is a consistent subset of $I$, then it is obvious that $V(I)\setminus S$ and $S$ are disjoint, and $V(I)\setminus S$ is covered by $S$ (see definitions in Section \ref{intro}).
\end{observation}
\begin{observation}\label{observation3}
For any $i\in \{0,1,\dots,2n\}$, $s=(i...i+1)$ is a leaf bar since $I_s=\emptyset$, $O_s=V(I)$, and $V(I)$ itself is a consistent subset.
\end{observation}
\begin{observation}\label{lemma1b}
For any leaf bar $s$, $O_s$ is a consistent subset of $I$ since $I=I_s\cup O_s$, $I_s \cap O_s=\emptyset$, and $I_s$ is covered by $O_s$.
\end{observation}
\noindent First, we find a subset $Z_s\subseteq O_s$ for each leaf bar $s$ such that $\lvert Z_s\rvert \leq 2\alpha$. $Z_s$ is called the \emph{useful cover} of the leaf bar $s$.
\subsection{Identifying Leaf Bars and the Useful Cover of Leaf Bars}
For points $0$ to $2n+1$, we create a $(2n+2)\times (2n+2)$ Boolean matrix named $M$ (see Table \ref{fig1}(B) corresponding to Figure \ref{fig1}(A)). For each pair $i$, $j$, $i<j$, we check whether $(i..j)$ is a leaf bar or not. If it's a leaf bar, then $M[i,j]=1$; otherwise, $M[i,j]=0$. It's basically a procedure for demonstrating how we calculate each leaf bar. For each $s=(i..j)$, it takes $\mathcal{O}(n^2)$ to verify whether $s$ is a leaf bar or not because we discover $\NN(I_s, O_s)$ in $\mathcal{O}(n^2)$ time and check whether $C(\NN(I_s, O_s))$ is equal to $C(I_s)$ or not in constant time. We use an Algorithm \ref{alg:b} to locate $Z_s$ when $s$ is a leaf bar, as shown below.

\SetKwComment{Comment}{/* }{ */}
\begin{algorithm}
\caption{An algorithm to find an efficient number $O(0..2n+1)$}\label{alg:b}
\KwData{An interval graph $I$ and corresponding color set $C=\{c_1, c_2, \dots, c_{\alpha}\}$. $s=(i..j)$ is a leaf bar, and $I_s \neq \emptyset.$}
\KwResult{ Useful cover $Z_s$ of $s$.}

     $Z_s\gets \emptyset$;
     
      $Q\gets \NN(I_s, O_s)$;
      Partition $Q=Q_L\cup Q_R\cup Q_O$;

      \If{$Q_O= \emptyset$}{
        \For{ each color $c_k\in C(I_s)$, $1\leq k\leq \alpha$ }{
        Find greatest left endpoint $t_1$ such that $I[t_1]\in Q_L$ and $C(I[t_1])=c_k$ (if exists);
        
        Find lowest right endpoint $t_2$ such that $I[t_2]\in Q_R$ and $C(I[t_2])=c_k$ (if exists);
        
        $Z_s \gets Z_s\cup \{I[t_1]\}\cup \{I[t_2]\}$;
        
        }
      
      }
      
      \If{$Q_O\neq \emptyset$}{
        \For{each color $c_k\in C(I_s)$, $1\leq k\leq \alpha$}{
        Find the largest left endpoint $t$ such that $I[t]\in Q_O$ and $C(I[t])=c_k$ if exists, and if it does not exist, use the above method for $Q_O= \emptyset$;
        
        $Z_s \gets Z_s\cup \{I[t]\}$ or $Z_s \gets Z_s\cup \{I[t_1]\}\cup \{I[t_2]\}$ depending on the two situations;
        
        }
      
      }
      
      Return $Z_s$;

\end{algorithm}

\noindent \begin{tcolorbox}[breakable,bicolor,
  colback=cyan!5,colframe=cyan!5,title=Interval Graph Construction.,boxrule=0pt,frame hidden]

\underline{\textbf{Finding $Z_s$ of Leaf Bar $s$:}}
 \vspace{5pt}
Given that $s=(i,j)$ is a leaf bar. Initially, $Z_s = \emptyset$. We assume $\lvert C(I_s)\rvert =k^{'}$ (see deinitions in Section \ref{intro}). We divide $Q$ into three distinct sets, $Q_L$, $Q_R$, and $Q_O$, as follows. $Q_L$ contains intervals where both ends are less than or equal to $i$. Similarly, intervals with both ends greater than or equal to $j$ are included in $Q_R$. In addition, $Q_O$ contains intervals with left endpoints less than or equal to $i$ and right endpoints greater than or equal to $j$. The cardinalities $\lvert Q_L\rvert $, $\lvert Q_R\rvert $, and $\lvert Q_O\rvert$ can range from $0$ to $\lvert Q\rvert$, depending on the condition of $s$. 

\smallskip

\noindent If $Q_O=\emptyset$ (see Figure \ref{fig1}(E)), then find the greatest left endpoint $t_1$ such that $I[t_1]\in Q_L$ and $C(I[t_1])=c_k$ (if exists) for each $k$, $1\leq k\leq \alpha$; that is, there exist at most $k^{'}$ such intervals. Find the lowest right endpoint $t_2$ such that $I[t_2]\in Q_R$ and $C(I[t_2])=c_k$ (if exists) for each $k$, $1\leq k\leq \alpha$; that is, there exist $k^{'}$ such intervals. We define $Z_s \gets Z_s\cup \{I[t_1]\}\cup \{I[t_2]\}$. Thus, $\lvert Z_s\rvert \leq 2k^{'} \leq 2\alpha$, and we are done.

\smallskip

\noindent If $Q_O\neq \emptyset$ , we have two instances. 1) If $C(Q_O)\geq k^{'}$ (see Figure \ref{fig1}(C)), then for each color $c_k$ in $C(I_s)$, we choose a greatest left endpoint $t$ such that $I[t]\in Q_O$ and $C(I[t])=c_k$. We define $Z_s\gets Z_s\cup \{I[t]\}$. Therefore, $\lvert Z_s\rvert \leq \alpha$, and we are done. 2) If $C(Q_O) < k^{'}$ (see Figure \ref{fig1}(D)), then for each color $c_k$ in $C(Q_O)$, we choose the greatest left endpoint $t$ such that $I[t]\in Q_O$ and $C(I[t])=c_k$. We define $Z_s\gets Z_s\cup \{I[i]\}$. For the remaining colors of $C(I_s)$, we use the procedure outlined above when $Q_O= \emptyset$. Thus, $\lvert Z_s\rvert = 2k^{'} \leq 2\alpha$, and we're done.
\end{tcolorbox}

\begin{lemma}\label{lemmaimportant}
If $s=(i..j)$ is a leaf bar, then the corresponding $Z_s$ of size at most $2\alpha$ produced by Algorithm \ref{alg:b} covers $I_s$. 
\end{lemma}
\begin{proof}
\noindent Given that $s$ is a leaf cover, then $I_s$ is covered by $O_s$. We will show that $I_s$ is covered by $Z_s$. Assume $I[i] \in I_s$. We've got two cases. Assume that $\lvert Q_O\rvert= 0$ (see Figure \ref{fig1}(E)). Since $s$ is a leaf bar, $I[i]$ must have a neighbor interval of $Q_L$ or $Q_R$. According to Algorithm \ref{alg:b}, we took the greatest left endpoint $t_1$ such that $I[t_1]\in Q_L$ and $C(I[t_1])=c_k$ for each such $c_k\in C(I_s)$ (if exists) in $Z_s$, as well as the lowest right endpoint $t_2$ such that $I[t_2]\in Q_R$ and $C(I[t_2])=c_k$ for each such $c_k\in C(I_s)$ (if exists) in $Z_s$. As a result, if $I[i]=c_j$, then $c_j\in C(I_s)$, which means that $I[i]$ has a neighbor interval of the same color in $Q_L\cup Q_R$. Hence, $I[i]$ is covered.

\smallskip

\noindent Assume that $\lvert Q_O\rvert \neq 0$. Assume $I[k]\in I_s$ is an interval with color $c_1$. If $c_1\in C(Q_O)$ (see Figure \ref{fig1}(C)), $I[k]$ is covered by an interval of $Q_O$, since according to Algorithm \ref{alg:b}, we picked the greatest left endpoint $t$ such that $I[t]\in Q_O$ and $C(I[t])=c_k$ (if exists) in $Z_s$. Here, $I[t]\in Q_O$ exists, and $C(I[t])=c_1$. If there is no interval $I[t]\in Q_O$ such that $C(I[t])=c_1$ (see Figure \ref{fig1}(D)), we repeat the prior situation when $\lvert Q_O\rvert= 0$ and $I[k]$ must be covered by an interval in $Q_L\cup Q_R$ that also appears in $Z_s$. Hence, our lemma is proven.
\end{proof}

\noindent The Algorithm \ref{alg:b} asserts that the size of $Z_s$ is at most $2\alpha$, and this applies to interval graphs. We must use this Algorithm \ref{alg:b} to prove our main result. We define \emph{leaf bar cover} as follows.
\begin{mdframed}[backgroundcolor=red!10,topline=false,bottomline=false,leftline=false,rightline=false] 
\begin{definition}\label{def}
A set of bars $D=\{s_1,\dots,s_m\}$ is a \emph{leaf bar cover} if the three conditions listed below are met.
\begin{enumerate}
    \item  $\cup_{i=1}^{m}I_{s_i}$ is covered by $\cup_{i=1}^{m}Z_{s_i}$ (see definitions in SEction \ref{intro} and observation \ref{easyobserve}).
    \item $V(I)=(\cup_{i=1}^{m}I_{s_i})\cup (\cup_{i=1}^{m}Z_{s_i})$.
    \item $(\cup_{i=1}^{m}I_{s_i})\cap (\cup_{i=1}^{m}Z_{s_i})=\emptyset$.
\end{enumerate}
\end{definition}
\end{mdframed}
\noindent So, a leaf bar cover "covers" the entire interval graph. We will eventually use the bars to create a minimum-sized leaf bar cover. The following lemma explains why we call a leaf bar cover.

\begin{figure}
    \centering
    \includegraphics[width=0.85\textwidth]{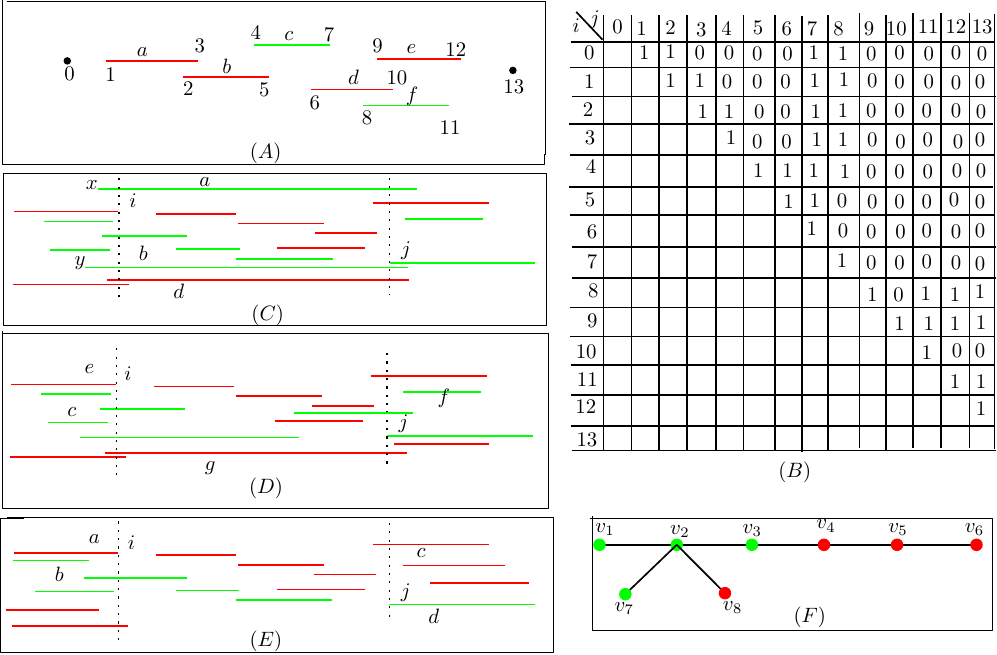}
    \caption{(A): The bar $s=(1..7)$ is a leaf bar. $l(s)=2$ and $r(s)=6$. $I[1]=I[3]=a$. (B): Boolean table M. $M[0,8]=1$, $M[0,9]=0$ because $(0..8)$ is a leaf bar but $(0..9)$ is not. (C): The 'green' color intervals in $Q_O$ are $a$ and $b$, and the greatest left endpoint is $x$. $I[x]=a$. The 'red' color interval in $Q_O$ is $d$. So, $Z_s=\{a,d\}$. (E): For 'red' color, $a$ and $b$ are the greatest left-endpoint interval and lowest right-endpoint interval in $Q_L$ and $Q_R$, respectively. Similarly, $c$ and $f$ are for 'green'. $Z_s=\{a,b,c,d\}$. (D): combining (C) and (E), we get $Z_s=\{c,f,g\}$.  (F): $\{v_3, v_4, v_8\}$ is a \cs and $\{v_7, v_8\}$ is an \mcs. $\{v_1, v_8\}$ is also an \mcs.} \label{fig1}
\end{figure}

\begin{lemma}\label{lemma2}
Let $S\subseteq I$ represent a consistent subset of $I$ where $\lvert S\rvert =d$ and no proper subset of $S$ is a consistent subset. Let $i_1 < i_2 < \dots < i_{2d}$ denote the points in $P(S)$. Let $s_j=(i_{j}..i_{j+1})$ for any $j$ with $1\leq j \leq 2d-1$. If $i_1>1$, take $s_0=(0..i_1)$; if $i_{2d}<2n$, take $s_{2d}=(2d..2n+1)$. $\{s_0,s_1, \dots,s_{2d}\}$ represents a leaf bar cover.
\end{lemma}
\begin{proof}
As $S$ covers $V(I)\setminus S$, it can be used to cover each leaf bar $s_j$ for any $j$, $0\leq j \leq 2d$. Therefore, we take $Z_{s_j}\subseteq S$ because $S$ covers each $I_{s_{j}}$. In addition, $\cup_{i=0}^{2d}Z_{s_i}=S$ because if $\cup_{i=0}^{2d}Z_{s_i}\subset S$, then a proper subset of $S$ is a consistent subset, which contradicts itself. We need to meet all three conditions of the definition \ref{def}. The first condition of the definition \ref{def} is satisfied by $\cup_{i=0}^{2d}Z_{s_i}=S$ and $\cup_{i=0}^{2d}I_{s_i}$ is covered by $S$. Note that $\cup_{i=0}^{2d}I_{s_i})=V(I)\setminus S$ and $\cup_{i=0}^{2d}Z_{s_i}=S$. Thus, $V(I)=(\cup_{i=1}^{m}I_{s_i})\cup (\cup_{i=1}^{m}Z_{s_i})$, and $(\cup_{i=1}^{m}I_{s_i})\cap (\cup_{i=1}^{m}Z_{s_i})=\emptyset$. Therefore, the second and third conditions of definition \ref{def} are met.
\end{proof}
Consistent subsets of $I$, including a \mcs, result in a leaf bar cover twice the size of the consistent subset. Our goal is to calculate an optimum leaf bar cover and obtain a consistent subset of "small" size.
\subsection{Computing an Optimal Leaf Bar Cover}\label{subsection2.1}
We define $O(0,2n+1)$ as the leaf bar cover that minimizes the number of bars in $O(0,2n+1)$; that is, it minimizes $\lvert O(0,2n+1)\rvert$. In \ref{alg:a}, we present a dynamic algorithm that yields $O(0,2n+1)$. We employ the right-to-left swift approach, as shown below. It is worth noting that if a bar $s$ is a leaf bar, then we use $s\in L$. It should be noted that $L$ is not a set but rather a notation that denotes whether $s$ is a leaf bar or not. If $s\in L$, it indicates a leaf bar. The symbol $\oplus$ represents the concatenation operator for sequences.
\noindent \begin{tcolorbox}[breakable,bicolor,
  colback=cyan!5,colframe=cyan!5,title=Interval Graph Construction.,boxrule=0pt,frame hidden]

\underline{\textbf{Dynamic Programming:}}\label{dp}
 \vspace{5pt}
To begin dynamic programming, we utilize the sets $X=\emptyset$, $Y=\emptyset$, and $R=\emptyset$. The right-to-left swift technique aims to identify a bar $s=(k..l)$ such that $X\cup Z_s$ covers $Y\cup I_s$ (see definitions in Section \ref{intro}), with the condition that the remaining part $\lvert O(0,k) \rvert$ becomes minimized for all feasible values of $k$ and $l$. As we explore right-to-left swift, it is important to examine whether all intervals having at least one endpoint in the set $\{1,\dots, k+1,k+2,\dots,2n\}$ appear in the set $X\cup Z_s \cup Y\cup I_s$ or not. If yes, we must terminate the recursion steps of the dynamic programming; otherwise, we pick $X\gets X\cup Z_s$ and $Y\gets Y\cup I_s$. That is, $X$ represents the consistent subset that covers $Y$. Note that $X$ and $Y$ must be disjoint; otherwise, it does not satisfy the property of a consistent subset (see Observation \ref{easyobserve}). By meeting all of the conditions listed above, we must dynamically discover $\lvert O(0,k)\rvert $ so that it becomes minimized for all possible values of $k$ and $l$. We make $R\gets R\cup \{s\}$, which means that $R$ becomes the leaf bar cover, and it will become an optimal leaf bar cover after finishing all the recursions of the dynamic programming \ref{alg:a}.
\end{tcolorbox}
\noindent \begin{tcolorbox}[breakable,bicolor,
  colback=cyan!5,colframe=cyan!5,title=Interval Graph Construction.,boxrule=0pt,frame hidden]

\underline{\textbf{Explanation:}}
 \vspace{5pt}
Let $R=\{s_1,\dots, s_{m}\}$ be the set obtained by the above dynamic programming \ref{alg:a}. First, we show that $R$ is a leaf bar cover. Suppose that $R$ is not a leaf bar cover; that is, at least one of the three conditions of the definition \ref{def} is not satisfied. According to dynamic programming, every interval of $V(I)$ must appear in $X\cup Y$ because every $\lvert O(0,k)\rvert$ will be minimized. We also check the condition that all the intervals with at least one endpoint in the set $\{1,\dots, k + 1, k + 2,\dots, 2n\}$ appear in the set $X\cup Z_s \cup Y\cup I_s$ before starting the next iteration; otherwise, $X$ becomes $X\gets X\cup Z_s$ and $Y\gets Y\cup I_s$ (see dynamic programming \ref{dp}). So, the second condition of the  definition \ref{def} is satisfied. Also, in dynamic programming \ref{dp}, in every recursion, we use $X\cap Y=\emptyset$, which satisfies the third condition of the definition \ref{def}. If the first condition of the definition \ref{def} is not satisfied, then there exist an interval $I[j]\in  \cup_{i=0}^{d}I_{s_i}$ such that $I[j]$ is not covered by $\cup_{i=0}^{d}Z_{s_i}$, which is also a contradiction because $I_s$ is covered by $Z_s$. We will discuss the running time of this dynamic programming later.
\end{tcolorbox}

\SetKwComment{Comment}{/* }{ */}
\begin{algorithm}
\caption{An algorithm to find an optimal number $\lvert O(0..2n+1)\rvert$}\label{alg:a}
\KwData{An interval graph $I$ and set of bars $L$. Some empty sets $X=\emptyset$, $Y=\emptyset$, $R=\emptyset$.}
\KwResult{ $\lvert O(0,2n+1)\rvert \leq 2\lvert OPT\rvert +1$. }

$N \gets n$\;
\For{$l \gets 1$ to $2n+1$}{
  \For{$k \gets 0$ to $l$}{
   
     \If{$s=(k..l)\in L$ and $I_{(k..l)}\neq \emptyset$}{
     Find $Z_s$ from Algorithm \ref{alg:b};
     
      \If{$X\cup Z_s$ covers $Y\cup I_s$}{
       $X_1 \gets X\cup Z_s$; $Y_1 \gets Y\cup I_s$;
       
       \For{$i \gets k+1$ to $2n$}{
        \If{$I[i]\in X_1\cup Y_1$ and $X_1\cap Y_1=\emptyset$}{
          $\lvert O(0,2n+1)\rvert \gets \lvert O(x,k)\rvert+1$;
          $R\gets R\cup \{s\}$;
        $X \gets X\cup Z_s$; $Y \gets Y\cup I_s$;
        }
     }
   
   }
   }
  
  }
}
Return $\lvert O(0,2n+1)\rvert$;

Return $R$ (note that $R$ is actually $O(0,2n+1)$);

\end{algorithm}

\begin{lemma}\label{lemma3}
The above dynamic programming \ref{alg:a} produces an optimal leaf bar cover.
\end{lemma}
\begin{proof}
Let $k<l$ be two integers in the dynamic programming \ref{alg:a}, such that $(k..l)\in L$ and $\lvert O(0,k)\rvert$ is minimized over all $k,j$, $0\leq k<l\leq 2n+1$. Clearly, $O(0,k)\oplus (k..j)$ is a leaf bar cover. Since $k$ minimizes $\lvert O(0,k)\rvert$, $\lvert O(0,2n+1)\rvert=\lvert O(0,k)\oplus (k..j)\rvert=\lvert O(0,k)\rvert +1$, which is given via dynamic programming.
\end{proof}

\noindent The following lemma states that this leaf bar cover is sufficient for our requirements.
\begin{lemma}\label{lemma2.3}
Let \mcs be the minimum consistent subset of the given interval graph. Therefore, $\lvert O(1,2n)\rvert \leq 2\lvert \mcs \rvert +1$.
\end{lemma}
\begin{proof}
Let $d$ be the size of \mcs and assume that $i_1, i_2, \dots, i_{2d}$ are the points in $P(\mcs)$. No proper subset of \mcs is a consistent subset. Using Lemma \ref{lemma2}, we show that $s_j=(i_{j}..i_{j+1})$ for any $j$ with $1\leq j \leq 2d-1$. If $i_1>1$, take $s_0=(0..i_1)$; if $i_{2d}<2n$, take $s_{2d}=(2d..2n+1)$. Therefore, $S=\{s_0,s_1, \dots,s_{2d}\}$ represents a leaf bar cover. By definition, $O(1,2n)$ is the smallest leaf bar cover, therefore $\lvert O(1,2n)\rvert \leq \lvert S\rvert \leq 2d+1$.
\end{proof}
\subsection{Computing a Consistent Subet from a Leaf Bar Cover}\label{subsection2.4}
Using the structural property of leaf bar cover in earlier sub-sections, we compute a consistent subset of $I$ that is within $2\alpha$ times the size of $O(1,2n)$. Since $O(1,2n)$ is roughly within twice the size of a minimum consistent subset, we will now demonstrate a $(4\alpha+2)$-approximation algorithm. We apply an Algorithm \ref{alg:c} to find a consistent subset called \acs.

\noindent \begin{tcolorbox}[breakable,bicolor,
  colback=cyan!5,colframe=cyan!5,title=Interval Graph Construction.,boxrule=0pt,frame hidden]

\underline{\textbf{Algorithm for Finding \acs:}}
 \vspace{5pt}
At first, we use $\acs = \emptyset$. We start finding $O(0,2n+1)$ in the dynamic programming \ref{alg:a}. In \ref{alg:a}, finding $O(0,2n+1)$ yields $\lvert O(0,2n+1)\rvert=\lvert O(i,k)\oplus (k..j)\rvert=\lvert O(i,k)\rvert+1$, where $s=(k..j)$ represents a leaf bar. When $s$ is a leaf bar, we apply Algorithm \ref{alg:b} to determine $Z_s$. We modify $\acs \gets \acs \cup Z_s$. We repeat this approach until we get $O(0,2n+1)$ in dynamic programming \ref{alg:a}.
\end{tcolorbox}
\SetKwComment{Comment}{/* }{ */}
\begin{algorithm}
\caption{Algorithm for Finding \acs}\label{alg:c}
\KwData{An interval graph $I$, and color set $C=\{c_1, c_2, \dots, c_{\alpha}\}$. $\acs=\emptyset$.}
\KwResult{ A consistent subset \acs.}
\underline{\textbf{step 1:}} Start finding $O(0,2n+1)$ using dynamic programming \ref{alg:a};

\underline{\textbf{step 2:}} We find $Z_{s}$ where $s=(k..j)$ where $\lvert O(0,2n+1)\rvert=\lvert O(i,k)\oplus (k..j)\rvert=\lvert O(i,k)\rvert +1$ while using Algorithm \ref{alg:b} in \underline{\textbf{step 1}}; $\acs \gets \acs \cup Z_s$;

\underline{\textbf{step 3:}} Update \underline{\textbf{step 1}} until we get $O(0,2n+1)$ using dynamic programming \ref{alg:a}.

Return \acs.
\end{algorithm}
\begin{theorem}
\acs is a consistent subset of the interval graph $I$.
\end{theorem}
\begin{proof}
\noindent Consider $I[i]\in V(I)\setminus \acs$. We argue that $I[i]$ is covered by \acs. First, we prove that $I[i]$ must occur in a leaf bar formed by $O(0,2n+1)$ in dynamic programming \ref{alg:a}. As $I[i]\in V(I)\setminus \acs$, $I[i]\notin \acs$. \acs is the union of every $Z_s$ (see Algorithm \ref{alg:c}), where $s$ is a leaf bar created via dynamic programming \ref{alg:a}. Suppose $s_0 , s_1, \dots, s_{m}$ are the total leaf bar created by $O(0,2n+1)$, then according to the definition of $O(0,2n+1)$ (see definition \ref{def}), $V(I)=(\cup_{j=0}^{m}I_{s_j})\cup (\cup_{j=0}^{m}Z_{s_j})$ and $(\cup_{j=0}^{m}I_{s_j})\cap (\cup_{j=0}^{m}Z_{s_j})=\emptyset$. Also, $\acs = \cup_{j=0}^{m}Z_{s_j}$. Thus, $I[i]\in \cup_{j=0}^{m}I_{s_j}$, implying that $I[i]$ must occur in a leaf bar produced by $O(0,2n+1)$ in dynamic programming \ref{alg:a}. We suppose $I[i]\in I_s$, thus $I[i]$ must be covered by $Z_s \subseteq \acs$. Hence, $I[i]$ is covered by \acs. $I[i]$ is arbitrary, hence \acs is a consistent subset.
\end{proof}

\subsection{Analysis of the Algorithm}

\begin{theorem}
The Algorithm \ref{alg:c} computes the consistent subset \acs with the cardinality $\lvert \acs\rvert \leq (4\alpha+2) \lvert \mcs\rvert$, where \mcs is the minimum consistent subset of $I$.
\end{theorem}
\begin{proof}
\noindent As previously stated, for each leaf bar $s$ generated by dynamic programming \ref{alg:a}, $Z_s \subseteq ACS$. Using Algorithm \ref{alg:b}, we obtain $\lvert \acs \rvert \leq 2\alpha \lvert O(0,2n+1) \rvert$. Using Lemma \ref{lemma2.3}, we get $\lvert \acs \rvert \leq 2\alpha \lvert O(0,2n+1) \rvert\leq 2\alpha (2\lvert \mcs \rvert +1) \leq (4\alpha +2)\lvert \mcs\rvert$ (because $\lvert \mcs\rvert \geq \alpha$ using observation \ref{basicobservation}).
\end{proof}
\begin{theorem}
The running time to find \acs in the Algorithm \ref{alg:c} is $\mathcal{O}(n^9)$.
\end{theorem}
\begin{proof}
The \underline{\textbf{step 1}} of the Algorithm \ref{alg:c} follows dynamic programming \ref{alg:a}. We assume any arbitrary number for $l$ and $k$, which takes $\mathcal{O}(n^2)$. To determine whether $(k..l)$ is a leaf bar or not ((line $4$ in dynamic programming \ref{alg:a}), we require $\mathcal{O}(n^2)$ time because for each interval $I[i]$ of $I_{(k..l)}$, we must locate $\NN(I[i],O_{(k..l)})$ and check if there is an interval in $\NN(I[i],O_{(k..l)})$ of color $C(I[i])$ or not. We now need to find $Z_s$ using Algorithm \ref{alg:b} (line $5$ in dynamic programming \ref{alg:a}).
\smallskip

\noindent We obtain $\NN(I_s, O_s)$ (line $2$ in Algorithm \ref{alg:b}) in $\mathcal{O}(n^2)$ using the same argument as above, and the rest of the Algorithm \ref{alg:b} takes $\mathcal{O}(n+2\cdot n)$ time. To find $Z_s$, we apply the Algorithm \ref{alg:b}, which takes $\mathcal{O}(n^2+n+2\cdot n)=\mathcal{O}(n^2)$ time.

\smallskip

\noindent To determine whether $X\cup Z_s$ covers $Y\cup I_s$ (line $6$ in dynamic programming \ref{alg:a}), we need $\mathcal{O}(n^2)$ and follow the same procedure as previously. Next, we need to check whether $I[i]\in X_1\cup Y_1$ and $X_1\cap Y_1= \emptyset$ or not for each $i$, $k+1\leq i\leq 2n$. This takes $\mathcal{O}(n)$ time. \underline{\textbf{step 2}} and \underline{\textbf{step 3}} are a recursive version of \underline{\textbf{step 1}}. The Algorithm \ref{alg:c} has a running time of $\mathcal{O}(n^2\cdot n^2 \cdot n^2 \cdot n^2 \cdot n)=\mathcal{O}(n^9)$ when all steps are combined. 
\end{proof}
\section{\apx-hardness of \mcs for Circle Graphs}\label{section4}
 \noindent The definitions of circle graphs and the dominating set of a graph are given below.
\begin{definition}
\begin{mdframed}[backgroundcolor=red!10,topline=false,bottomline=false,leftline=false,rightline=false]\label{def4} 

 A graph $T$ is a circle graph if there is a one-to-one correspondence between vertices in $V(T)$ and a set $X$ of chords in a circle such that two vertices in $V(T)$ are adjacent if and only if the corresponding chords in $X$ intersect. The \emph{dominating set} of a graph $G$ is a subset $S\subseteq V(G)$ such that for each vertex $v\in V(G)$, $\neigh[v]\cap S \neq \emptyset$.
\end{mdframed}
\end{definition}
\noindent We obtain a gap-preserving reduction from the minimum dominating set problem in circle graphs, which is known to be \apx-hard \cite{circle}. Consider the geometric representation of a circle graph $T$ with $\lvert V(T)\rvert = n$ chords.  Taking a closer look at the \apx-hardness proof \cite{circle}, we can suppose that no two chords share an endpoint; that is, there are exactly $2n$ distinct points on the circle determining the endpoints of the chords. We reduce an instance of a dominating set in a circle graph $T$ (with $\lvert V(T) \rvert=n$ and $\lvert E(T) \rvert =m$) to an instance of a dominating set of a circle graph $T^{'}$, as follows.

\begin{figure}
    \centering
    \includegraphics[width=0.85\textwidth]{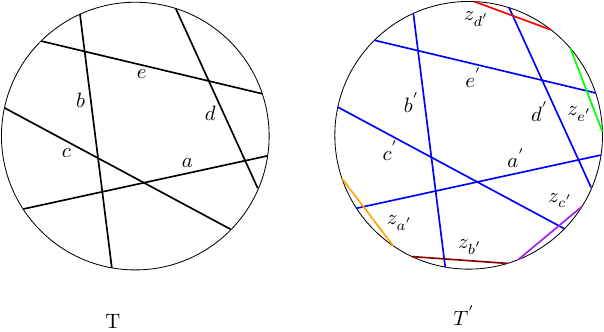}
    \caption{An example reduction}
    \label{fig:enter-label}
\end{figure}

\noindent \begin{tcolorbox}[breakable,bicolor,
  colback=cyan!5,colframe=cyan!5,title=Interval Graph Construction.,boxrule=0pt,frame hidden]

\underline{\textbf{Circle Graph Construction:}}
 \vspace{5pt}
We put a new chord $v^{'}$ in the set $V(T^{'})$ corresponding to each chord $v\in V(T)$. We include these chords in the set $V_1$. Therefore, $\lvert V_1\rvert =n$. We arrange the chords of $V_1$ in $T^{'}$ so that if any two chords $u$ and $v$ of $V(T)$ intersect, then the corresponding two chords $u^{'}$ and $v^{'}$ of $V_{1}$ are also intersect in $T^{'}$. For each chord $v^{'}\in V_{1}$, we establish a new chord $z_{v^{'}}$ in a set $V_{2}$ adjacent to $v^{'}$ only (see Figure \ref{fig:enter-label}). So, $\lvert V_2\rvert=n$. Each chord of $V_1$ has the color $c_1$, while each chord of $V_2$ has different colors that differ from $c_1$. We suppose that the color set of $V_2$ is $\{c_2, c_3, \dots, c_{n+1}\}$. This is a complete reduction. Therefore, $V(T^{'})=V_1\cup V_2$, $\lvert V(T^{'})\rvert =2n$, and $\lvert E(T^{'})\rvert =n+m$. Additionally, the color set of $T^{'}$ is $\{c_1, c_2, \dots, c_{n+1}\}$.
\end{tcolorbox}
\begin{lemma}{\label{lemma1_mscs_apx}}
$T$ has a dominating set of size $k$ if and only if $T^{'}$ has a consistent subset of size $n+k$.
\end{lemma}
\begin{proof}
$(\Rightarrow)$ 
\noindent Let $A$ be a dominating set of $T$ with size $k$. We're creating a consistent subset $B$ of $T^{'}$ as follow. Based on observation \ref{basicobservation}, $V_2\subseteq B$. As a result, we included all of the chords from $V_2$ in $B$. Assume $A=\{v_1, v_2, \dots, v_k\}$. We put $v_i^{'}$ in $ B$ for each $v_i\in A$. It is obvious that $B$ has at least one chord of each color from the collection $\{c_1, c_2, \dots, c_{n+1}\}$. Furthermore, for any vertex $v_l^{'}\in V(T^{'})\setminus B$, there exists a chord in $B$ of the type $v_i^{'}$ such that $\dist (v_l^{'}, v_i^{'})=1$; otherwise, the corresponding set $A$ cannot be a dominating set. Thus, $B$ is a consistent subset of size $n+k$.

\smallskip

\noindent $(\Leftarrow)$ Assume $B$ is a consistent subset of size $n + k$. We are creating a dominating set $A\subseteq V(T)$ where $\lvert A\rvert =k$. As per observation \ref{basicobservation}, $B$ must have at least one chord of each color. Each chord of the type $z_{v^{'}}$ has a different color, and this is the only chord in $T^{'}$ that has color $C(z_{v^{'}})$, hence $V_2 \subseteq B$. Consider $B\setminus V_2=\{v_1^{'}, v_2^{'}, \dots, v_{k}^{'}\}$. We claim that the chords in $T$ corresponding to the chords $B\setminus V_2$ in $T^{'}$ forms a dominating set, that is, $A=\{v_1, v_2,\dots,v_k\}$ is a dominating set. Assuming $A$ is not a dominating set, there exists $v_i\in V(T)\setminus A$ such that $\dist(v_i, A)>1$. Thus, $B$ is not a consistent subset because $\dist(v_i^{'}, B)>1$, $\dist(v_i^{'},z_{v_i^{'}})=1$, and $C(v_i^{'})\neq C(z_{v_i^{'}})$, which is a contradiction.
\end{proof}
\begin{theorem}
    The \mcs problem in circle graphs is \apx-hard.
\end{theorem}
\begin{proof}
    Using a similar gap-preserving reduction of the paper \cite{circle} and the lemma \ref{lemma1_mscs_apx}, we conclude that \mcs is \apx-hard in circle graphs.
\end{proof}

\section{Remarks}   
This study introduces the first approximation algorithm for \mcs, as there is yet no approximation algorithm to date. Therefore, it is also very important to show the approximation algorithms in various graph classes. As we prove \apx-hard result for \mcs in circle graphs, an approximation algorithm is also required. \fpt can be used as an open problem for interval and circle graphs in the future.

\bibliographystyle{unsrt} 
\bibliography{references}  

\begin{thebibliography}{10}

\bibitem{Hart}
Peter~E. Hart.
\newblock The condensed nearest neighbor rule (corresp.).
\newblock {\em {IEEE} Trans. Inf. Theory}, 14(3):515--516, 1968.

\bibitem{Wilfong}
Gordon~T. Wilfong.
\newblock Nearest neighbor problems.
\newblock In Robert L.~Scot Drysdale, editor, {\em Proceedings of the Seventh Annual Symposium on Computational Geometry, North Conway, NH, USA, , June 10-12, 1991}, pages 224--233. {ACM}, 1991.

\bibitem{Bodhayan}
Kamyar Khodamoradi, Ramesh Krishnamurti, and Bodhayan Roy.
\newblock Consistent subset problem with two labels.
\newblock In B.~S. Panda and Partha~P. Goswami, editors, {\em Algorithms and Discrete Applied Mathematics - 4th International Conference, {CALDAM} 2018, Guwahati, India, February 15-17, 2018, Proceedings}, volume 10743 of {\em Lecture Notes in Computer Science}, pages 131--142. Springer, 2018.

\bibitem{Chitnis}
Rajesh Chitnis.
\newblock Refined lower bounds for nearest neighbor condensation.
\newblock In Sanjoy Dasgupta and Nika Haghtalab, editors, {\em Proceedings of The 33rd International Conference on Algorithmic Learning Theory}, volume 167 of {\em Proceedings of Machine Learning Research}, pages 262--281. PMLR, 29 Mar--01 Apr 2022.

\bibitem{Biniaz}
Ahmad Biniaz, Sergio Cabello, Paz Carmi, Jean{-}Lou~De Carufel, Anil Maheshwari, Saeed Mehrabi, and Michiel H.~M. Smid.
\newblock On the minimum consistent subset problem.
\newblock In Zachary Friggstad, J{\"{o}}rg{-}R{\"{u}}diger Sack, and Mohammad~R. Salavatipour, editors, {\em Algorithms and Data Structures - 16th International Symposium, {WADS} 2019, Edmonton, AB, Canada, August 5-7, 2019, Proceedings}, volume 11646 of {\em Lecture Notes in Computer Science}, pages 155--167. Springer, 2019.

\bibitem{diestel2012graph}
Reinhard Diestel.
\newblock Graph theory, volume 173 of.
\newblock {\em Graduate texts in mathematics}, page~7, 2012.

\bibitem{Banerjee}
Sandip Banerjee, Sujoy Bhore, and Rajesh Chitnis.
\newblock Algorithms and hardness results for nearest neighbor problems in bicolored point sets.
\newblock In Michael~A. Bender, Martin Farach{-}Colton, and Miguel~A. Mosteiro, editors, {\em {LATIN} 2018: Theoretical Informatics - 13th Latin American Symposium, Buenos Aires, Argentina, April 16-19, 2018, Proceedings}, volume 10807 of {\em Lecture Notes in Computer Science}, pages 80--93. Springer, 2018.

\bibitem{Sanjana}
Sanjana Dey, Anil Maheshwari, and Subhas~C. Nandy.
\newblock Minimum consistent subset of simple graph classes.
\newblock In Apurva Mudgal and C.~R. Subramanian, editors, {\em Algorithms and Discrete Applied Mathematics - 7th International Conference, {CALDAM} 2021, Rupnagar, India, February 11-13, 2021, Proceedings}, volume 12601 of {\em Lecture Notes in Computer Science}, pages 471--484. Springer, 2021.

\bibitem{Anil}
Sanjana Dey, Anil Maheshwari, and Subhas~C. Nandy.
\newblock Minimum consistent subset problem for trees.
\newblock In Evripidis Bampis and Aris Pagourtzis, editors, {\em Fundamentals of Computation Theory - 23rd International Symposium, {FCT} 2021, Athens, Greece, September 12-15, 2021, Proceedings}, volume 12867 of {\em Lecture Notes in Computer Science}, pages 204--216. Springer, 2021.

\bibitem{Arunima}
Hiroki Arimura, Tatsuya Gima, Yasuaki Kobayashi, Hiroomi Nochide, and Yota Otachi.
\newblock Minimum consistent subset for trees revisited.
\newblock {\em CoRR}, abs/2305.07259, 2023.

\bibitem{Bubai}
Bubai Manna and Bodhayan Roy.
\newblock Some results on minimum consistent subsets of trees.
\newblock {\em CoRR}, abs/2303.02337, 2023.

\bibitem{bubai1}
Aritra Banik, Sayani Das, Anil Maheshwari, Bubai Manna, Subhas~C Nandy, Bodhayan Roy, Sasanka Roy, Abhishek Sahu, et~al.
\newblock Minimum consistent subset in trees and interval graphs.
\newblock {\em arXiv preprint arXiv:2404.15487}, 2024.

\bibitem{circle}
Mirela Damian and Sriram~V. Pemmaraju.
\newblock Apx-hardness of domination problems in circle graphs.
\newblock {\em Information Processing Letters}, 97(6):231--237, 2006.

\end{thebibliography}
\end{document}